\documentclass[12pt]{iopart}

\usepackage{amssymb,amsthm,graphicx}
\usepackage{graphicx}
\usepackage{colordvi}
\usepackage{color}
\usepackage{iopams}

\newcommand{\N}{\mathbb{N}}
\newcommand{\R}{\mathbb{R}}

\newcommand{\Hm}[1]{\leavevmode{\marginpar{\tiny%
$\hbox to 0mm{\hspace*{-0.5mm}$\leftarrow$\hss}%
\vcenter{\vrule depth 0.1mm height 0.1mm width \the\marginparwidth}%
\hbox to
0mm{\hss$\rightarrow$\hspace*{-0.5mm}}$\\\relax\raggedright #1}}}

\newcommand{\sgn}{\mathop{\mathrm{sgn}}\nolimits}

\newtheorem{claim}{Claim}[section]
\newtheorem{theorem}[claim]{Theorem}
\newtheorem{lemma}[claim]{Lemma}

\newtheorem{proposition}[claim]{Proposition}

\theoremstyle{definition}
\newtheorem{remark}[claim]{Remark}
\newtheorem{remarks}[claim]{Remarks}

\begin{document}


\title[Gap asymptotics in a weakly bent leaky quantum wire]
{Gap asymptotics in a weakly bent leaky quantum wire}

\author{Pavel Exner}
\address{Doppler Institute for Mathematical Physics and Applied
Mathematics, \\ Czech Technical University in Prague,
B\v{r}ehov\'{a} 7, 11519 Prague, \\ and  Nuclear Physics Institute
CAS, 25068 \v{R}e\v{z} near Prague, Czechia} \ead{exner@ujf.cas.cz}

\author{Sylwia Kondej}
\address{Institute of Physics, University of Zielona G\'ora, ul.\ Szafrana
4a, 65246 Zielona G\'ora, Poland} \ead{s.kondej@if.uz.zgora.pl}

\maketitle
\begin{abstract}
\noindent The main question studied in this paper concerns the weak-coupling behavior of the geometrically induced bound states of singular Schr\"odinger operators with an attractive $\delta$ interaction supported by a planar, asymptotically straight curve $\Gamma$. We demonstrate that if $\Gamma$ is only slightly bent or weakly deformed, then there is a single eigenvalue and the gap between it and the continuum threshold is in the leading order proportional to the fourth power of the bending angle, or the deformation parameter. For comparison, we analyze the behavior of a general geometrical induced eigenvalue in the situation when one of the curve asymptotes is wiggled.
\end{abstract}

\section{Introduction}\label{introduction}

The paper is devoted to an asymptotic problem for a class of singular Schr\"{o}dinger operators, usually called leaky quantum wires, or more generally graphs. They are used to model electron motion in thin wires or networks made of semiconductor or other materials. In contrast to the conventional quantum graph models \cite{BK} they employ a confinement mechanism which does not neglect the quantum tunneling. At the same time, these operators pose various new and interesting questions in the spectral geometry. They have been an object of intense interest for more than a decade --- for a survey of results up to 2008 see, e.g., the review paper \cite{Ex08} --- but there are still numerous open problems in this area.

To describe the subject of the present paper, let us first characterize the operators in question which play the role of Hamiltonians in such systems. They can be formally written as
\begin{equation}\label{eq-formal}
-\Delta -\alpha \delta_\Gamma\,,
\end{equation}
where $\Delta $ is the Laplace operator acting in $L^2(\R^n)$ and $\delta_\Gamma $ stands for the Dirac-type potential supported by a manifold $\Gamma$ of a lower dimensionality; we will be concerned with the particular situation when $n=2$ and $\Gamma$ is a curve in the plane. The above formal expression is sometimes written also as $-\Delta -\alpha \delta(x-\Gamma)$ or $-\Delta -\alpha (\delta_\Gamma,\cdot)\delta_\Gamma$; a proper mathematical definition of the operator, which we will denote\footnote{In most papers dealing with this subject the symbol $H_{\alpha,\Gamma}$ is used, however, the coupling constant $\alpha$ is fixed here and for the sake of simplicity we drop it.} as $H_\Gamma$, will be given below. Note also that as long as the curve is smooth one can define $H_\Gamma$ alternatively through boundary conditions describing a jump of the normal derivative across $\Gamma$; this gives it an illustrative meaning of a $\delta$-interaction perpendicular to the curve \cite{AGHH}.

Here we are going to deal with the situation when $\Gamma$ is a single infinite curve. If it is a straight line, $\Gamma=\Sigma$, the spectrum is easily found by separation of variables: it is absolutely continuous and
$$
\sigma(H_{\Sigma}) = \left[-\textstyle{\frac14}{\alpha^2}, \infty \right)\,.
$$
Once the curve becomes geometrically nontrivial, the spectrum changes. In particular, if $\Gamma$ is asymptotically straight at both end in an appropriate sense, then the essential spectrum remain preserved, $\sigma_\mathrm{ess}(H_{\Gamma}) = \left[-\frac14{\alpha^2}, \infty \right)$, however, one or more isolated eigenvalues appear below its threshold \cite{EI}. Relations between between properties of this discrete spectrum and the geometry of the curve are of a great interest.

As an example, consider a `broken' line $\Gamma_\beta$ which consists of two halflines meeting at a `vertex' and forming the exterior angle $\beta \in [0, \pi]$, in other words, for $\beta =0$ the curve $\Gamma_0$ coincides with the straight line $\Sigma $. There are various recent results on the discrete spectrum dependence on the angle $\beta$, in particular, asymptotic bounds on the number of eigenvalues for $\beta \to\pi-$ \cite{DR} or a lower bound on the principal eigenvalue \cite{Lo}. Our problem here concerns the situation when $\beta$ is small. Then there is a single eigenvalue and we ask how the gap between it and the essential spectrum threshold behaves as $\beta\to 0+$. We shall not restrict at that to this simple example and consider a class of locally bent (or deformed) curves, straight outside a compact region, which become straight when the corresponding parameter analogous to the angle $\beta$ approaches zero.

This is one of the open problems formulated in the  review \cite{Ex08} where the conjecture was made that the gap is proportional to the fourth power of the angle parameter in the leading order. This guess comes from an analogy with quantum waveguides. The Dirichlet Laplacian in a bent infinite planar strip of a fixed width $d$ has also geometrically induced eigenvalues below the threshold of its essential spectrum, equal to $\left( \frac{\pi}{d} \right)^2$ in this case, and the gap is proportional to the fourth power of the bending angle \cite{DE}. However, analogies are often treacherous guides in these situations, and the present one might not work due to several reasons. First of all, the type of confinement is rather
different in the two situations, the present one being much `softer'. Secondly, one does not use the same method to deal with the weak bending asymptotics. In the waveguide case the eigenvalue existence is proved variationally and the asymptotics uses the Birman-Schwinger method adapted from the standard Schr\"odinger operator theory. Here the (generalized) Birman-Schwinger trick was  used already to prove the existence \cite{EI} and the only natural way to proceed is to analyze finer properties of the BS operator. Nevertheless, our main goal in this paper is to demonstrate that the mentioned conjecture was correct, namely that we have the asymptotic relation
$$
\lambda (H_{\Gamma_\beta}) =  -\frac14{\alpha^2} +a \beta^4 + o(\beta^4)
$$
as $\beta\to 0+$  with a coefficient $a<0$ the explicit form of which will be given in Theorem~\ref{th-main} below. A similar relation is proved for more general weakly bent curves.

Passing from the example to a more general class, a caveat is needed: it is vital we consider geometric perturbations (bending, local deformation) which disappear in the limit and $\Gamma$ becomes a straight line. To underscore this fact we analyze also the situation in which the unperturbed curve is a line locally deformed in a compact region, denoted as $\overline{\Gamma }$. The operator $H_{\overline{\Gamma}}$ has then an eigenvalue, or eigenvalues below $-\frac14{\alpha^2}$. If we now wiggle one of the halfline `ends' introducing a nonzero angle $\varphi$ between it an the other one, the eigenvalue dependence on the angle is linear in the leading correction term, 
$$
\lambda (H_{\overline{\Gamma}_\varphi}) = \lambda
(H_{\overline{\Gamma}}) + b \varphi + o(\varphi)
$$
with the coefficient $b$ given explicitly in Theorem~\ref{th-main2}. The result extends to the situation when the unperturbed eigenvalue is degenerate and we present the corresponding formula even if it might be void: one can conjecture that the spectrum of the operators $H_\Gamma$ considered here is simple.

\section{Preliminaries and the results}
\setcounter{equation}{0}

\subsection{Geometry of the interaction support}

Let $\Gamma $ be a continuous and piecewise $C^2$ infinite planar curve without self-intersections which we parametrize naturally by its arc length.
More precisely, we suppose that $\Gamma $ is the graph of a piecewise $C^2$ function $\gamma:\, \R \to \R^2 $ where the argument $s$ of $\gamma (\cdot )$ determines the length of arc, that is, $|\dot\gamma(s)|=1$. With a later purpose in mind we also  introduce a special symbol $\Sigma $ for a
straight line. Moreover, we assume that

\begin{description}

\item[$\mathbf{H_1}$] there exists a $c\in (0, 1)$ such that\footnote{In the following we use the generic symbol $c,c',\dots$ for the various positive constants}
\begin{equation}\label{eq-Cbound}
|\gamma (s)-\gamma (s')|\geq c|s-s'| \quad \mathrm{for}\,\,\mathrm{any} \,\, s,s' \in \R\,,
\end{equation}

\item[$\mathbf{H_2}$] there are real numbers $s_1 >s_2 $ and straight lines $\Sigma_i,\: i=1,2,$ such that $\Gamma $ coincides with $\Sigma_1 $ for the parameter values $s\geq s_1$ and with $\Sigma_2 $ for $s\leq  s_2$,

\item[$\mathbf{H_3}$] one-side limits of $\dot\gamma$ exists at the points where the function $\ddot\gamma$ is discontinuous.

\end{description}

\noindent The first assumption excludes, in particular, the existence of cusps, self-intersections as well as what one could call `near self-intersections', that is, it guarantees the existence of strip neighborhood of $\Gamma$ that does not intersect itself. The second one says that $\Gamma $ is straight outside a compact region. We also introduce a special symbol $\overline{\Gamma} $ for curves for which the two asymptotes mentioned above are parts of the same line, $\Sigma_1 =\Sigma_2$.

The third assumption means that the signed curvature\footnote{The subscripts in this formula refer to the Cartesian coordinates, not to be mixed with the deformation parameter $\beta$ of $\gamma_\beta$ used in the following.} $k(s)=\dot\gamma_2(s)\ddot\gamma_1(s) - \dot\gamma_1(s)\ddot\gamma_2(s)$, where the dot conventionally denotes the derivative with respect to $s$, is piecewise continuous and the one-sided limits of $\dot\gamma$, that is of the tangent vector to the curve, at the points of discontinuity exist. We denote the set of these points as $\Pi = \{p_i \}_{i=1}^{\sharp \Pi}$ and shall speak of them as of vertices, having in mind the interpretation of $\Gamma$ as a chain graph. Consequently, $\Gamma $ consists of  $\sharp \Pi +1$ simple arcs or edges (by assumption their number is finite, and two of them are infinite), each of which has as its endpoints one or two of the vertices.

We note that the curvature integral describes \emph{bending} of the curve. More specifically, the bending of $\Gamma $ between two points $\gamma
(s)$ and $\gamma (s')$ away of the vertices, understood as the angle between the tangents at those points, equals
\begin{equation}\label{eq-totcurvature}
  \phi (s,s') = \sum _{p_i\in(s,s')} c(p_i)
  +\int_{(s,s')\setminus\Pi} k(s)\,\mathrm{d}s\,,
\end{equation}
where $c(p_i)\in (0,\pi)$ is the exterior angle formed by the two adjacent edges of $\Gamma $ which meet at $p_i$, in other words, $\cos c(p_i) = t_{i,+} \cdot  t_{i,-} $, where $t_{i,\pm} = \lim _{s\to s_i\pm} \dot\gamma(s)$ are the unit tangent vectors at the endpoints of the two edges
which meet at $p_i$. Alternatively, we can understand $\phi (s,s')$ as the integral over the interval $(s,s')$ of $\tilde{k}:\: \tilde{k}(s) = k(s) + \sum_{p\in\Pi} c(p)\,\delta(s-p)$. By assumption the functions $k,\,\tilde{k}$ are compactly supported, thus $\phi(s,s')$ has the same value for all $s<s_2 <s_1<s'$ which we shall call the \emph{total bending}. It is easy to check that the curvature, in the generalized sense of $\tilde{k}$, allows us to reconstruct the curve uniquely up to Euclidean transformations using the formul{\ae}
\begin{equation}\label{eq-reprgamma}
  \gamma (s) = \left( \int_{0}^s \cos \phi (u, 0 )\,\mathrm{d}u
  \,,\int_{0}^s \sin \phi (u, 0 )\,\mathrm{d}u \right)\,,
\end{equation}
identifying without loss of generality the point $\gamma(0)$ with the origin of coordinates in the plane and $\dot\gamma(0)$ with the vector $(1,0)$.

\subsection{Definition of the Hamiltonian}

We define our singular Schr\"odinger operator as the form sum of the free and the interaction parts. Note that the trace map theorem implies that the embedding $W^{1,2}(\R^2 ) \hookrightarrow L^2(\Gamma )$ is continuous, and consider the following form,
$$
\mathrm{h}_\Gamma  [f]= \int_{\R^2}|\nabla f(x)|^2\,\mathrm{d}x -\alpha
\int_{\Gamma }|f(s)|^2\, \mathrm{d}s\,, \quad f\in W^{1,2}(\R^2)\,,
$$
where the function in the second term on the right-hand-side is understood in the sense of the mentioned trace map. The form is closed and bounded from below, hence it is by the second representation theorem associated with a unique self-adjoint operator which we denote by $H_\Gamma $, giving thus a rigorous meaning to the formal expression (\ref{eq-formal}). As was mentioned in the introduction (and elsewhere), it can be alternatively defined through the boundary conditions imposed at the curve giving an illustrative meaning of an attractive $\delta$ potential of the strength $\alpha$ supported by the curve $\Gamma$ to the interaction.

Since the curve $\Gamma $ belongs under the assumptions we made into the class analyzed in Ref.~\cite{EI}, we know that
\begin{equation} \label{spess}
\sigma_{\mathrm{ess}} (H_\Gamma )=\sigma_{\mathrm{ess}} (H_\Sigma)
=\left[-\textstyle{\frac14}{\alpha^2}, \infty \right)\,,
\end{equation}
and that $H_\Gamma $ has at least one discrete eigenvalue whenever $\Gamma \neq \Sigma$.

\subsection{Main result}

As indicated in the introduction, our main result concerns the gap behavior in the situation when the geometric perturbation is weak. First we have to make more precise the meaning of the asymptotics. We start from a fixed curve $\Gamma$ satisfying the assumptions ${\bf H_1}$--${\bf H_3}$, to which a bending function corresponds  to (\ref{eq-totcurvature}), and consider the one-parameter family of `scaled' curves $\Gamma_\beta $ described by the functions
\begin{equation}\label{eq-reprgamma}
  \gamma _\beta (s) = \left(\int_{0}^s \cos \beta \phi (u, 0)\,\mathrm{d}u
  \,,\int_{0}^s \sin \beta \phi (u, 0))\,\mathrm{d}u \right)\,, \quad |\beta|\in(0,1]\,.
\end{equation}
Note that the limit $\beta\to 0+$ which we are interested in may have a different meaning depending of the value of the total bending of $\Gamma$. If the latter is nonzero, the curves $\Gamma_\beta$ are `straightening' as $\beta\to 0+$. Our considerations include also the situation when the total bending is zero and $\Gamma$ is a local deformation of the straight line, then the limit can be regarded rather a `flattening' of the deformation. In both cases, of course, we arrive at the straight line for $\beta=0$.

To formulate the main result we need  an additional notion and an auxiliary statement. We define an integral operator $A:\, L^2 (\R)\to L^2 (\R)$ through its kernel,
\begin{equation}\label{eq-defA}
 \mathcal{A}(s,s'):= \frac{\alpha^4}{32\pi} K_0'\left( \frac{\alpha }{2}|s-s'|
\right)\left( |s-s'|^{-1} \left( \int_{s'}^s  \phi \right)^2 -
\int_{s'}^s \phi^2 \right)\,,
\end{equation}
where we use the abbreviation $\int_{s'}^s  \phi = \int_{s'}^s \phi ( u, s' )\mathrm{d}u$, etc., and $K_0' (\cdot)$ stands for the
derivative of the MacDonald function $K_0 (\cdot)$; alternatively we can express the kernel without the derivative using the relation $K'_0(u) = -K_1(u)$, cf.~\cite{AS, GR}. In Sec.~\ref{ss:prooflemma} below we shall show that the kernel (\ref{eq-defA}) is doubly integrable:

\begin{lemma} \label{le-traceclassA}
Under the stated assumptions,
\begin{equation}\label{eq-Abound}
\int_{\R\times\R}\mathcal{A}(s,s')\,\mathrm{d}s \,\mathrm{d}s '< \infty\,.
\end{equation}
\end{lemma}

\noindent Now we are in position to state our main result.

\begin{theorem}\label{th-main}
There is a $\beta_0 >0$ such that for any $\beta \in (-\beta_0,0) \cup (0,\beta_0 )$ the operator $H_{\Gamma_\beta }$ has a unique eigenvalue
$\lambda (H_{\Gamma_\beta })$ which admits the asymptotic expansion
\begin{equation}\label{eq-asymbeta}
\lambda (H_{\Gamma_\beta })= -\frac{\alpha^2}{4}
-\left(\int_{\R\times \R
}\mathcal{A}(s,s')\,\mathrm{d}s\,\mathrm{d}s' \right)^2 \beta^4
+o(\beta^4)\,.
\end{equation}
\end{theorem}

\begin{remarks} \emph{(a) The nature of the first correction  term:} to understand the nature of the leading term in the gap expansion expressed by
Theorem~\ref{th-main} we note that it is comes from two sources, as will be made clear below: from the difference of Birman-Schwinger kernels referring to $\Gamma_\beta $ from $\Gamma $, measured by
$$
K_0 \left(\frac{\alpha}{2}|\gamma_\beta  (s)-\gamma_\beta (s')|\right)- K_0\left(\frac{\alpha}{2}|s-s'|\right)\,
$$
and from the measure of shortening the Euclidean distance of the curve points relative to that measured along the curve arc,
$$
|\gamma_\beta (s) -\gamma_\beta (s')|- |s-s'|\,.
$$
In this context the appearance of $K_0'$ in (\ref{eq-defA}) is natural, while the second factor entering the coefficient in (\ref{eq-defA}) reflects the behavior of the above geometric quantity. \\[.5em]
\emph{(b) The broken line example:} Let us specify the asymptotics for the simple example we have used as a motivation in the introduction. Given $\beta  \in [0, \pi )$ we define $\Gamma_\beta  :=\{ (x,y)\in \R^2 \,:\, y=0\,\, \mathrm{ for}\,\, x\leq 0\,\, \wedge
\,\,  y=x \tan \beta  \,\, \mathrm{ for}\,\, x>  0 \}$. The arc length parameterization yields
$$ \gamma_\beta  (s)= \left\{ \begin{array}{ll}
(s\,\cos\beta, s\,\sin\beta )  & \mathrm{for}\; s>0\,  \\[.5em]
(s,0)  & \mathrm{otherwise}\,. \end{array} \right.
$$
We introduce the notation
$$
\Omega =\{(s,s')\in \R^2 \,:\, \sgn s = -\sgn s' \} \quad \mathrm{
and}\quad  \Omega^c =  \R^2 \setminus \Omega \,,
$$
for the union of the open second and fourth quadrant, and its complement. A straightforward calculation then leads to
\begin{equation}\label{eq-defA1}
 \mathcal{A}(s,s')= \frac{\alpha^4}{32\pi} K_0'\left( \frac{\alpha }{2}|s-s'|
\right)\frac{|ss'|}{|s-s'|}\,\chi_{\Omega}(s,s')\,,
\end{equation}
where $\chi_{\Omega}(\cdot ,\cdot)$ is the characteristic function of $\Omega$. Note that coefficient is in this case proportional to $\alpha ^2$ dependence, since
$$
\int_{\R\times \R }\mathcal{A}(s,s')\,\mathrm{d}s\,\mathrm{d}s' =
\frac{\alpha }{4\pi } \int_{\Omega} K_0'\left( |s-s'|
\right)\frac{|ss'|}{|s-s'|}\,\mathrm{d}s\,\mathrm{d}s'\,.
$$
This is what one expects, because the broken line is self-similar object and a change of the coupling constant is equivalent to a scaling transformation. Note that the last integral can be calculated explicitly. Namely, 
\begin{eqnarray} \nonumber
\int_{\Omega} K_0'\left( |s-s'|
\right)\frac{|ss'|}{|s-s'|}\,\mathrm{d}s\,\mathrm{d}s'=
2   \int_{\R_+ \times \R_+ }K_0'\left( s+s'
\right)\frac{ss'}{s+s'}\,\mathrm{d}s\,\mathrm{d}s' = \\ \nonumber -\left( \int_{0}^{\infty }K_1 (t)t^2 \mathrm{d}t \right)
\left(  \int_0 ^{\pi /2} \frac{\sin 2 \psi }{(\cos \psi +\sin \psi )^4}  \mathrm{d}\psi \right) = -\frac{2}{3}\,.
\end{eqnarray}
Consequently, the spectral gap is given by $\frac{\alpha^2}{36\pi^2} \beta^4$, or in the relative expression
$$
\frac{-\frac14\alpha^2 - \lambda (H_{\Gamma_\beta})}{-\frac14\alpha^2} = - \frac{1}{9\pi^2}\beta^4 + o(\beta^4)\,.
$$
\end{remarks}

\section{Wiggling one halfline} \label{sec-mildII}
\setcounter{equation}{0}

In the introduction we warned that the above result holds only in the true weak-coupling limit. To make this caveat more illustrative, let us now discuss the case when the unperturbed `zero-angle' curve has at least one isolated eigenvalue. Consider a planar curve $\overline{\Gamma}$ being a local deformation of the straight line which satisfies the assumptions ${\bf H_1}$--${\bf H_3}$. Its arc-length parameterization is described by the function $\overline{\gamma}:\, \R \to \R^2 $. Now we perturb it: for $\varphi \in (0, \pi)$ we construct the curve $\overline{\Gamma}_\varphi$ which is the graph of
$$ \overline{\gamma}_\varphi  (s)= \left\{ \begin{array}{ll}
(s\, \cos  \varphi , s\, \sin \varphi )  & \mathrm{for}\;\, s>s_0 \,,  \\[.5em]
 \overline{\gamma}(s)  & \mathrm{otherwise}\,,
  \end{array} \right.
$$
where $s_0 \geq s_1$ is a suitable point. Without loss of generality we may assume $s_0= 0$ since we have a `translational' freedom in the choice of the parameterization. The deformation $\overline{\Gamma}_\varphi  $ of $\overline{\Gamma}$ thus means wiggling the `right' halfline end by the angle $\varphi$. The set
$$
\Omega=\{(s,s')\in \R^2 \,:\, \sgn s=- \sgn s'\}
$$
can be decomposed as follows,
$$
\Omega=\Omega^I\cup \Omega^{\mathit{II}}\,,
$$
where $\Omega^I:= \{(s,s')\in \Omega \,:\, s  \leq 0\}$ an $\Omega^{\mathit{II}}:=\Omega\setminus \Omega^{\mathit{I}}$.

If $H_{\overline{\Gamma }}$ has more than one eigenvalue we arrange them in the ascending order. Suppose that $\lambda (H_{\overline{\Gamma }})= \lambda_{k} (H_{\overline{\Gamma }}),\: k=j+1, \dots ,j+m$ is such an eigenvalue of the multiplicity $m$. In fact, one conjectures that $m=1$ but since the simplicity of the spectrum has not been proven, we consider degenerate eigenvalues as well. To state the result we employ an integral operator $A_k$
in $L^2 (\R)$ with the kernel
$$
\mathcal{A}_k(s,s')= \breve{\mathcal{A}}_k
(s,s')\chi_{\Omega^{I}}+ \breve{\mathcal{A}}_k
(s',s)\chi_{\Omega^{\mathit{II}}}\,,
$$
where
$$
\breve{\mathcal{A}}_k (s,s')= - \frac{\alpha^2 \kappa_k}{2\pi}\,
\frac{s\,\gamma_2 (s')}{\overline{\rho } (s,s')}\, K_0 ' \big(\kappa_k
\overline{\rho} (s,s') \big)\,,\quad \kappa_k := \sqrt{- \lambda_k
(H_{\overline{\Gamma } })  }
$$
and
$$
\overline{\rho }(s,s')=|\gamma (s)-\gamma (s')|\,.
$$
Furthermore, $\{f_k\}_{k=j+1}^{j+m}$ stands for the corresponding Birman-Schwinger eigenfunctions, see (\ref{eq-ef2}) below.

\begin{theorem}\label{th-main2}
The eigenvalue $\lambda (H_{\overline{\Gamma}})$ splits under the perturbation, in general, into $m$ eigevalues of $H_{\overline{\Gamma
}_\varphi}$ and  the following asymptotic expansion
\begin{equation}\label{eq-mainev}
\lambda_k (H_{\overline{\Gamma } _\phi}) = \lambda_k
(H_{\overline{\Gamma } }) +(A_k f_k , f_k )_{L^2(\R)}\varphi
+o(\varphi)\,.
\end{equation}
is valid for $k=j+1, \dots , j+m$.
\end{theorem}

\section{Proof of Theorem~\ref{th-main}} \label{s:mainproof}
\setcounter{equation}{0}

In this section we demonstrate our main result. First we check that the leading-order coefficient is well defined, then we prove the theorem itself using a refined Birman-Schwinger-type argument.

\subsection{Proof of Lemma~\ref{le-traceclassA} } \label{ss:prooflemma}

There are two parts of the $(s,s')$ plane where the behavior of the kernel has to be checked. Let us first inspect how $\mathcal{A}(s,s')$ for looks like as $|s-s'|\to 0$. Since the bending is bounded by assumption, $|\phi (s,s')|\leq c$ for any $s,s' \in \R$, one can easily check that
\begin{equation}\label{eq-estim1a}
0 \leq -\left( |s-s'|^{-1} \left( \int_{s'}^s  \phi \right)^2 -
\int_{s'}^s \phi^2 \right) \leq c'|s-s'|\,.
\end{equation}
On the other hand, the function $\R_+ \ni x \mapsto K_0' (x) $ is continuous behaving as $K'_0 (x)\sim x^{-1}$ in the limit $x\to 0+$, cf.~\cite{AS}. Combining this with (\ref{eq-estim1a}) we conclude that the kernel is bounded, $|\mathcal{A}(s,s')| \leq c''$, and consequently, it has no singularities and the convergence of the integral is determined by the decay of $\mathcal{A}(s,s')$ at infinity.

Far from the origin we use the fact the $\Gamma $ is there straight so that $\phi (s,s') = 0 $ holds  for $(s,s')\in \Omega_\mathrm{asympt}:= (-\infty , s_2)\times (-\infty , s_2)\cup (s_1 ,\infty ) \times (s_1, \infty )$, hence in view of the definition (\ref{eq-defA}), $\mathcal{A}(s,s')$ vanishes at this set as well. Take now that $s,s'\in \Omega \setminus \Omega_\mathrm{asympt}$ and suppose $|s- s'|\to \infty$. It is easy to see that far from the origin we have $|s-s'|=|s|+|s'|$, and using the asymptotics $K'_0 (x)\sim \mathrm{e}^{-x} x^{-1/2}$ for $x\to +\infty$ together with (\ref{eq-estim1a}) one obtains $| \mathcal{A}(s,s')| \leq c \mathrm{e}^{-\alpha(|s|+|s'|)/2}(|s|+|s'|)$. This yields
\begin{equation}\label{eq-Abound1}
\int_{\R\times \R}\mathcal{A}(s,s')\,\mathrm{d}s \,\mathrm{d}s'< \infty\,,
\end{equation}
completing thus the proof of Lemma~\ref{le-traceclassA}.

\begin{remark}
Note that $\mathcal{A}(s,s')\geq 0$ holds for any $s,s'\in \R$. Indeed, we have $K'_0(x) <0$ for any $x >0$ and the bracketed expression in  the definition (\ref{eq-defA}) is also negative, cf.~(\ref{eq-estim1a}).
\end{remark}

\subsection{Modification of the Birman-Schwinger argument}

Fix $\kappa>0$. It is well known that the resolvent $G_\kappa =(-\Delta +\kappa^2)^{-1}\,:\, L^2(\R^2) \to L^2(\R^2)$ is an integral operator with
the kernel
$$
\mathcal{G}_\kappa (x,y) = \frac{1}{2\pi} K_0 (\kappa |x-y|)\,.
$$
Relying on the general results of \cite{BEKS} we can define a bilateral embedding of $G_k$ into the space $L^2 (\Gamma ) \cong L^2 (\R)$. Specifically, let $Q_{\Gamma_\beta } (\kappa )$ stand for an integral operator with the kernel
\begin{equation}\label{eq-kernelQ}
\mathcal{Q}_{\Gamma_\beta } (\kappa; s,s')=\frac{1 }{2\pi} K_0
(\kappa |\gamma_\beta (s)-\gamma_\beta  (s')|)\,.
\end{equation}
The generalized Birman-Schwinger principle then reads
\begin{equation}\label{eq-BS}
  -\kappa^2 \in \sigma _{\mathrm{d}}(H_{\Gamma_\beta }) \quad \Leftrightarrow
  \quad \ker (I-\alpha Q_{\Gamma_\beta  }(\kappa )) \neq
   \emptyset\,,
\end{equation}
and moreover,
$$
\dim \ker (H_{\Gamma_\beta } + \kappa^2) = \dim \ker (I-\alpha
Q_{\Gamma_\beta  }(\kappa ))\,.
$$
cf.~\cite[Cor.~2.3]{BEKS}. Since $\sigma _{\mathrm{disc}}(H_\beta) \subset (\infty, -\frac14\alpha^2)$ holds in view of (\ref{spess}), we can restrict ourselves to $\kappa \in (\frac12\alpha, \infty )$. Since the straight line Birman-Schwinger operator $I-\alpha Q_{\Sigma }(\kappa )$ is invertible for $\kappa>\frac12\alpha$ we obtain
$$
I-\alpha Q_{\Gamma_\beta  }(\kappa )=(I-\alpha Q_{\Sigma }(\kappa
)) \left[I - \alpha (I-\alpha Q_{\Sigma }(\kappa ))^{-1}
(Q_{\Gamma_\beta }(\kappa )- Q_{\Sigma }(\kappa ))\right]\,.
$$
Combining the above equivalence with (\ref{eq-BS}) we conclude that $ -\kappa^2 \in \sigma _{\mathrm{d}}(H_\phi)$ holds \emph{iff}
\begin{equation}\label{eq-BS1}
\ker \left[I - \alpha (I-\alpha Q_{\Sigma }(\kappa ))^{-1}
(Q_{\Gamma_\beta }(\kappa )- Q_{\Sigma }(\kappa ))\right] \neq
\emptyset\,.
\end{equation}
We denote
$$
\epsilon = \beta^2 \quad \mathrm{and } \quad  \kappa_\delta =
\sqrt{\textstyle{\frac14}{\alpha ^2}+\delta^2}\,,\;\;\delta>0\,.
$$
With these notations equation (\ref{eq-BS1}) reads
\begin{equation}\label{eq-BS2}
\ker \left[I - B_\delta D_\epsilon (\kappa_\delta )\right] \neq
\emptyset\,,
\end{equation}
where
\begin{equation}\label{eq-DB}
D_{\epsilon }(\kappa ): =\alpha \left( Q_{\Gamma_\beta }(\kappa )-
Q_{\Sigma }(\kappa )\right) \quad \mathrm{and }\quad B_\delta
:= (I-\alpha Q_{\Sigma }(\kappa_\delta ))^{-1}\,.
\end{equation}
To make use of the equivalence (\ref{eq-BS}) we need an auxiliary regularization. Specifically, let $$V_{a} (s):= \mathrm{e}^{a |s|/4}\,$$ and put
$$
\hat{B}_\delta := V_{-\alpha }B_\delta V_{-\alpha }\,,\quad
\check{D}_\epsilon:=V_{\alpha }D_\epsilon  V_{\alpha }\,.
$$
Obviously, $\left[I - B_\delta D_\epsilon (\kappa_\delta )\right]f =0$ holds \emph{iff} $\left[I - \hat{B}_\delta \check{D}_\epsilon
(\kappa_\delta )\right]\hat{f} =0$, where $\hat{f} = V_{-\alpha }f$, and the Birman-Schwinger principle (\ref{eq-BS}) modifies to the form
\begin{equation}\label{eq-BS2}
  -\kappa_\delta^2 \in \sigma _{\mathrm{d}}(H_\beta ) \quad \Leftrightarrow
  \quad \ker (I-\hat{B}_\delta \check{D}_\epsilon
(\kappa_\delta )) \neq \emptyset\,,
\end{equation}
with the additional property that
$$
\dim \ker (H_{\beta } + \kappa_\delta^2) = \dim \ker
(I-\hat{B}_\delta \check{D}_\epsilon (\kappa_\delta ))\,;
$$
this will be the starting point for the further analysis of the discrete spectrum of $H_{\Gamma_\beta }$.

\subsection{Auxiliary results: asymptotics of $\check{D}_\epsilon (\kappa ) $ and $\hat{B}_\delta$}

The first needed statement concerns the asymptotics of $D_\epsilon(\kappa )$ for $\epsilon$ small.

\begin{lemma} \label{le-expD}
We have
\begin{equation}\label{eq-expD}
D_\epsilon (\kappa ) = D^{(1)}(\kappa ) \epsilon +D^{(2)}_\epsilon
(\kappa ) \epsilon^2 \,,
\end{equation}
where $D^{(1)} (\kappa )$ is an integral operator with the kernel
$$
\mathcal{D}^{(1)} (\kappa ; s,s') = -\frac{ \alpha \kappa }{4\pi
}\, K_0 '(\kappa |s-s'|)\left( |s-s'|^{-1} \left( \int_{s'}^s  \phi
\right)^2 - \int_{s'}^s \phi^2 \right) \,.
$$
Moreover, both $D^{(1)}(\kappa )$ and $D^{(2)}_\epsilon (\kappa )$ are Hilbert-Schmidt operators and the \textrm{HS}-norm of $D^{(2)}_\epsilon (\kappa )$ has a uniform bound w.r.t. $\epsilon$.
\end{lemma}
\begin{proof}
Denote
$$
\rho  (s,s'):=|\gamma_\beta  (s)-\gamma_\beta  (s')|\,,\quad
\sigma (s,s')=|s-s'|\,.
$$
First we expand $\rho$ with respect $\epsilon =\beta ^2$. To this aim we employ the formula
\begin{equation}\label{eq-estim3}
  \rho  (s,s')= \left( \left(\int_{s'}^s \cos \beta \phi ( u, s')\mathrm{d}u
  \right)^2+
 \left(\int_{0}^s \sin \beta \phi (u, s')\mathrm{d}u \right)^2\right)^{1/2}\,,
\end{equation}
which follows from (\ref{eq-reprgamma}). Using the mean value theorem together with the expansions $\sin\nu = \nu  - \frac{\nu^3}{3!}+\cdots$ and  $\cos \nu= 1 -\frac{\nu^2}{2!}+\cdots$ for $\nu $ small we conclude that there are $\theta_i = \theta_i (\epsilon) \in (0,1),\: i=1,2,$ such that
$$
\rho =\sigma +\left. \frac{\mathrm{d}\rho }{\mathrm{d}\epsilon }\right|_0 \epsilon
+\zeta_\epsilon (s,s')\epsilon^2\,,\qquad  \zeta_\epsilon
(s,s')= \left. \frac{\mathrm{d}^2\rho}{\mathrm{d}\epsilon^2 }\right|_{\theta_2
\epsilon}\,\theta_1\,,
$$
where the symbol $|_a$  conventionally means that the value of a  function is supported at the point $a$. Furthermore, a straightforward calculation shows that
\begin{equation}\label{eq-rhobound}
\hspace{-2em}\left. \frac{d\rho }{d\epsilon }\right|_0 = -\frac{1}{2}\left(
|s-s'|^{-1}\left( \int_{s'}^s \phi \right)^2 - \int_{s'}^s \phi^2
\right) \,, \quad \left|\zeta_\epsilon (s,s')\right| \leq
C|s-s'|\,.
\end{equation}
Consequently,
\begin{equation}\label{eq-estim1}
 \hspace{-2em} \rho (s,s')-\sigma  (s,s')=-\frac{1}{2}\left(
|s-s'|^{-1}\left( \int_{s'}^s \phi \right)^2 - \int_{s'}^s \phi^2
\right) \epsilon +
  \zeta_\epsilon
  (s,s')\epsilon^2\,
\end{equation}
Now we expand $K_0 (\kappa \rho )$. Employing again the mean value theorem we have
\begin{equation}\label{eq-expK}
\hspace{-4em} K_0 (\kappa \rho )=K_0 (\kappa \sigma  ) + \kappa K_0  ' (\kappa
\sigma  )(\rho -\sigma )+\kappa ^2  K_0  '' (\kappa
(\sigma+\check{\theta}_2 (\rho - \sigma ) ) )\check{\theta}_1(\rho
- \sigma  )^2
\end{equation}
with some $\check{\theta}_i = \check{\theta}_i (\epsilon )\in (0,1)$, $i=1,2$. Consequently, inserting (\ref{eq-estim1}) to (\ref{eq-expK}) we obtain 
\begin{equation}\label{eq-decompK}
K_0 (\kappa \rho )=K_0 (\kappa \sigma  ) +\mathcal{K}^{(1)}(s,s')
\,\epsilon + \mathcal{K}^{(2)}_\epsilon(s,s')\,\epsilon^2\,,
\end{equation}
where
$$
\mathcal{K}^{(1)}(s,s')=-\frac{\kappa }{2 }\,K_0 '(\kappa
|s-s'|)\left( |s-s'|^{-1}\left( \int_{s'}^s \phi \right)^2 -
\int_{s'}^s \phi^2 \right) \,,
$$
and
$$
\mathcal{K}^{(2)}_\epsilon(s,s') = \kappa K_0  ' (\kappa \sigma
)\,\zeta_\epsilon +\kappa ^2  K_0  '' (\kappa
(\sigma+\check{\theta}_2 (\rho - \sigma ) ) )\,\check{\theta}_1(\rho
- \sigma  )^2 \,.
$$
Using the asymptotics $K'_0 (x)\sim x^{-1}$ and $K''_0 (x)\sim x^{-2}$ for $x\to 0+$ together with the inequalities $\left| \left. \frac{\mathrm{d}\rho }{\mathrm{d}\epsilon }\right|_0 \right| \leq c|s-s'|$, $\:|\zeta_\epsilon(s,s')| \leq c|s-s'| $, and the equation (\ref{eq-estim1}) we verify that both $\mathcal{K}^{(1)}$ and $\mathcal{K}^{(2)}_\epsilon$ are bounded as $|s-s'|\to 0 $. Proceeding in the same way as in the proof of Lemma~\ref{le-traceclassA} we conclude that for $(s,s')\in \Omega_\mathrm{asympt} $ both $\mathcal{K}^{(1)}$ and $\mathcal{K}^{(2)}_\epsilon$ vanish. On the other hand due to the asymptotics of  $K'_0 (\cdot)$ and $K''_0 (\cdot )$ for $x\to +\infty$ we can estimate $\mathcal{K}^{(1)}$ and
$\mathcal{K}^{(2)}_\epsilon$ by
\begin{equation}\label{eq-boundK}
c\,e^{-\kappa (|s|+|s'|)}(|s|+|s'|)\,,
\end{equation}
for $s,s'\in \Omega \setminus \Omega_\mathrm{asympt}$ and $|s-s'|\to \infty$. This implies $\mathcal{K}^{(1)}\,, \mathcal{K}_\epsilon ^{(2)} \in L^2
(\R\times \R) $. Using now the equations (\ref{eq-DB}) and (\ref{eq-kernelQ}) with $\mathcal{D}^{(1)}(\kappa;\cdot ,\cdot )= \frac{\alpha }{2\pi} \mathcal{K}^{(1)}(\cdot,\cdot )$ and $\mathcal{D}_\epsilon^{(2)}(\kappa;\cdot ,\cdot )=\frac{\alpha}{2\pi} \mathcal{K}^{(2)}_\epsilon (\cdot ,\cdot )$ we get the claim of the lemma.
\end{proof}

\begin{remark} \label{re-conth}
\rm{With a latter purpose on mind we note that the function $h(\cdot  )=\int_{\R \times \R} \mathcal{D}^{(1)}(\cdot ; s,s')\,\mathrm{d}s \mathrm{d}s'$ is continuous on $(0,\infty)$. Indeed, for $|s-s'|$ small  we have
$$
\left| K_0 ' (\kappa |s-s'|) -K_0 ' (\kappa' |s-s'|) \right| \sim \frac{|\kappa -\kappa '|}{\kappa \kappa' |s-s'|}\,.
$$
On the other hand, for $s,s'\in \Omega\setminus \Omega_{\mathrm{asympt}}$ we can estimate
\begin{eqnarray*}
\lefteqn{\left| K_0 ' (\kappa |s-s'|) -K_0 ' (\kappa' |s-s'|) \right| \sim
\left|   \frac{\e^{-\kappa (|s|+|s'|)}-\e^{-\kappa' (|s|+|s'|)}}{\sqrt{|s|+|s'|}}\right|} \\[.5em] && \hspace{1em}
\leq |\kappa -\kappa '|\, \e^{-\kappa (|s|+|s'|)}\, (\sqrt{|s|+|s'|})\,.
\end{eqnarray*}
Using the arguments employed in the proof of the previous lemma one can claim that $h(\kappa')=h(\kappa)+\mathcal{O}(\kappa -\kappa')$.  }
\end{remark}

Note that the bound (\ref{eq-boundK}) has the following easy consequence.

\begin{proposition} \label{prop-1}
We have
\begin{equation}\label{eq-expD}
\check{D}_\epsilon (\kappa ) = \check{D}^{(1)}(\kappa ) \epsilon
+\check{D}^{(2)}_\epsilon (\kappa ) \epsilon^2 \,,
\end{equation}
where $\check{D}^{(1)}(\kappa ) = V_\alpha D^{(1)} (\kappa ) V_\alpha $ and $\check{D}_\epsilon^{(1)}(\kappa ) = V_\alpha D^{(1)}_\epsilon (\kappa ) V_\alpha $. Moreover, for any $\kappa >\frac12\alpha$ both $ \check{D}^{(1)}(\kappa )$ and $\check{D}^{(2)}_\epsilon (\kappa )$ are Hilbert-Schmidt.
\end{proposition}

Our next aim is to analyze the asymptotics of $B_\delta $ for $\delta $ small. In the following lemma we single out the  regular and singular components of this opearator.

\begin{lemma}
The operator $B_\delta$ admits the decomposition
\begin{equation}\label{eq-decomB-1}
B_\delta = S_{\delta  }+R_{\delta  }\,,
\end{equation}
where $S_{\delta  }$ is an integral operator with the kernel
$$
\frac{\alpha ^2}{4}\, \frac{\mathrm{e}^{-\delta |s-s'|}}{\delta }
$$
and $R_{\delta }$ is a bounded operator in $L^2 (\R)$ with a bound uniform w.r.t. $\delta$. 
\end{lemma}
\begin{proof} Let $\mathcal{F}$  denote the Fourier-Plancherel operator in $L^2 (\R)$. The BS operator $Q_\Sigma (\kappa )$ corresponding to the straight line is of convolution type, and consequently, $\mathcal{F}^{-1}Q_\Sigma (\kappa )\mathcal{F}$ acts as multiplication by $\frac{1}{2}\alpha (p^2+ \kappa^2)^{-1/2}$, cf.~\cite[Thm.~5.2]{EI}. This means that  $B_\delta =(I-Q_\Sigma (\kappa_\delta ))^{-1}$ is unitarily equivalent to multiplication by
$$
b _{\delta}(p):=
\frac{(p^2+\frac14\alpha^2+\delta^2)^{1/2}}{(p^2+\frac14\alpha^2+\delta^2)^{1/2}-\frac12\alpha}\,.
$$
A straightforward calculation shows that
\begin{equation}\label{eq-beta}
\hspace{6em} b _{\delta}(p)=
\frac{\alpha^2}{2}\frac{1}{p^2+\delta^2}+r_{\delta, \alpha }(p)\,,
\end{equation}
where
$$
r_{\delta, \alpha }(p)= 1+\frac{\alpha}{2}\,\frac{1}{(p^2+\frac14\alpha^2+\delta^2)^{1/2}+\frac12\alpha}\,.
$$
Using the bound $|r_{\delta }(p)|\leq c_\alpha $ uniform w.r.t. $\delta $ together with the equivalence
$$
\frac{1}{2\pi }\int_{\R} \frac{\mathrm{e}^{ipx }}{p^2 +\delta ^2}\,
\mathrm{d}p= \frac{1}{2}\,\frac{\mathrm{e}^{-\delta   |x|}}{2\delta  }\,,
$$
applying it to (\ref{eq-beta}) we get (\ref{eq-decomB-1}) with $R_{\delta }$ being a Fourier transform of the operator of multiplication by
$r_{\delta , \alpha }(p)$.
\end{proof}

The proved statement leads to a natural decomposition of the operator $\hat{B}_\delta$.

\begin{proposition} \label{prop-2}
The operator $\hat{B}_\delta$ admits the following decomposition
\begin{equation}\label{eq-decompoB2}
\hat{B}_\delta = \frac{1}{\delta }L +M_\delta\,,
\end{equation}
where $L$ is the one rank operator $L=\frac14{\alpha^2}( \cdot , V_{-\alpha} )V_{-\alpha }$ and $M_\delta $ is a Hilbert-Schmidt operator with the norm $\|M_\delta \|_{HS} $ bounded uniformly w.r.t. $\delta $.
\end{proposition}
\begin{proof} In view of the definition of $S_\delta $ we get
$$
\hat{S}_\delta=V_{-\alpha} S_\delta V_{-\alpha}= \frac{1}{\delta
}L +N_\delta\,,
$$
where the kernel of $N_\delta $ takes the following form
$$
\mathcal{N}_\delta (s,s')=\frac{\alpha ^2}{4\delta }\,\mathrm{e}^{-\alpha
|s|/4}(1-\mathrm{e}^{-\delta |s-s'|})\, \mathrm{e}^{-\alpha |s|/4}\,.
$$
The inequality  $|\mathcal{N}_\delta (s,s')|\leq \frac14{\alpha ^2}\,\mathrm{e}^{-\alpha |s|/4}|s-s'|\, \e^{-\alpha |s|/4}$ implies $\|N_\delta \|_{HS}\leq c$. This, in view of (\ref{eq-decomB-1}), proves (\ref{eq-decompoB2}) with $M_\delta=N_\delta + \hat{R}_\delta$, where $\hat{R}_\delta$ denotes $V_{-\alpha }R_\delta V_{-\alpha }$.
\end{proof}
 
The idea of decomposing the Birman-Schwinger operator into the sum of a rank-one singular operator and a regular remainder is well known and powerful tool in analysis of weak-coupling constant regular perturbations  \cite{S75}. It has been also used to treat Schr\"odinger operator with weak singular potentials, cf.~\cite{KL}.

\subsection{Concluding the proof of Theorem~\ref{th-main}.}

Propositions~\ref{prop-1} and~\ref{prop-2} allow us to write the following asymptotic expansion,
\begin{equation}\label{eq-final}
  \hat{B}_\delta \check{D} (\kappa_\delta )=
  \frac{\epsilon}{\delta} L \check{D}^{(1)}(\kappa_\delta )+A_{\delta ,
  \epsilon}\epsilon \,,
\end{equation}
where $\|A_{\delta,\epsilon}\|_{HS} \leq c$. The first term on the right-hand side of (\ref{eq-final}) is a rank-one  operator. Therefore,
$I-A_{\delta, \epsilon}\epsilon $ is  invertible for $\epsilon$ small enough and, consequently, we arrive at the expression
  $$
I -  \hat{B}_\delta \check{D} (\kappa_\delta )= (I-A_{\delta ,
  \epsilon}\epsilon )\left[ I - \frac{\epsilon}{\delta} (I-A_{\delta ,
  \epsilon}\epsilon )^{-1}  L \check{D}^{(1)}(\kappa_\delta
  )\right]\,.
  $$
Combining the above equation with  (\ref{eq-BS2})  we  can express the generalized Birman-Schwinger principle in the following way,
\begin{equation}\label{eq-final2}
 \ker \left[I - \frac{\epsilon}{\delta} (I-A_{\delta,
  \epsilon}\epsilon )^{-1}  L \check{D}^{(1)}(\kappa_\delta
  )\right]\neq \emptyset\,.
  \end{equation}
Since $(I-A_{\delta,\epsilon}\epsilon )^{-1}L \check{D}^{(1)}(\kappa_\delta)$ is a rank-one operator, the operator equation (\ref{eq-final2}) has a unique solution for $\epsilon $ small enough. Furthermore, making use of the expansion $(I-A_{\delta, \epsilon}\epsilon )^{-1} = I + A_{\delta, \epsilon}\epsilon + \cdots$ are able to we conclude that (\ref{eq-final2}) is equivalent to
\begin{equation}\label{final3}
  \delta = \epsilon \mathrm{Tr}\,\left[ L \check{D}^{(1)}(\kappa_\delta  )\right]
   + o(\epsilon) \,
\end{equation}
Finally, note that 
$$
\mathrm{Tr}\,\left[ L\check{D}^{(1)}(\kappa_\delta  )\right]= \frac{\alpha^2}{4}
(\check{D}^{(1)}(\kappa_\delta  )V_{-\alpha} , V_{-\alpha} )_{L^2 (\R)}=
\frac{\alpha^2}{4}     \int_{\R\times \R} \mathcal{D}^{(1)}(\kappa_\delta ; s,s') \mathrm{d}s
\mathrm{d}s'\,.
$$
Making use of Remark~\ref{re-conth} we conclude that
$$
\int_{\R\times \R} \mathcal{D}^{(1)}(\kappa_\delta ; s,s') \mathrm{d}s
\mathrm{d}s'=\int_{\R\times \R} \mathcal{D}^{(1)}(\textstyle{\frac12}\alpha  ; s,s') \mathrm{d}s
\mathrm{d}s'+o(1)\,.
$$
Combining the above facts we get
$$\mathrm{Tr}\,\left[ L\check{D}^{(1)}(\kappa_\delta  )\right]=
\frac{\alpha^2}{4}\int_{\R\times \R} \mathcal{D}^{(1)}(\textstyle{\frac12}\alpha  ; s,s') \mathrm{d}s
\mathrm{d}s'+o(1)=
\int_{\R\times \R} \mathcal{A}(s,s')\,\mathrm{d}s\,\mathrm{d}s'+o(1)\,;
$$
this, in view of (\ref{final3}), implies
\begin{equation}\label{eq-defdelta}
\delta = \left( \int_{\R\times \R}
\mathcal{A}(s,s')\,\mathrm{d}s\,\mathrm{d}s' \right) \, \epsilon
+o(\epsilon)\,.
\end{equation}
Using $\epsilon = \beta ^2$ and combining $\lambda (H_\phi) =-\kappa^2_\delta= -\frac14{\alpha^2}-\delta^2$ with (\ref{eq-defdelta}) we arrive finally at the claim of Theorem~\ref{th-main}.

\section{Proof of Theorem~\ref{th-main2}}
\setcounter{equation}{0}

In Sec.~\ref{sec-mildII} we have described the  geometry $\overline{\Gamma }$. The deformation $\overline{\Gamma }_\varphi$ of $\overline{\Gamma }$ is defined by means of
$$
\overline{\gamma }_\varphi  (s)= \left\{ \begin{array}{ll}
(s \cos  \varphi , s \sin \varphi )  & \mathrm{for }\; s>s_0   \\[.5em]
\overline{\gamma }(s)  & \mathrm{otherwise},
  \end{array} \right.
$$
for  $s_0 = 0$, chosen as the point from which the wiggling starts.

\subsection{Birman-Schwinger principle for $H_{\overline{\Gamma}_\phi }$}

To analyze the discrete spectrum behavior of $H_{\overline{\Gamma}_\phi} $ we will again employ the Birman-Schwinger principle, however, this time the unperturbed operator will not be determined by the Schr\"{o}dinger operator with the interaction supported by a straight line but rather by $H_{\overline{\Gamma} }$. This makes the problem essentially different and requires a detailed discussion.

Let us first look at the eigenvalues and eigenfunctions of $H_{\overline{\Gamma}}$. We denote by $\{\lambda_{n} (H_ {\overline{\Gamma}})\}_{n=1}^N$
the discrete eigenvalues of $H_{\overline{\Gamma}} $ and by $\{f_n\}_{n=1}^N$ determine the corresponding eigenfunctions; as remarked above we take into account a possible spectral degeneracy. Applying the generalized Birmann-Schwinger principle to this situation one gets the existence condition
\begin{equation}\label{eq-BS2a}
\ker (I-\alpha Q_{\overline{\Gamma}} (\kappa_n) ) \neq \emptyset
\,,\quad \mathrm{where }\quad \kappa_n=\sqrt{-\lambda_n
(H_{\overline{\Gamma} } )}\,,
\end{equation}
and moreover, the eigenfunction corresponding to $\lambda_n(H_{\overline{\Gamma}})$ is given by
\begin{equation}\label{eq-ef2}
g_n  = \mathcal{G}_{\kappa_n} \ast f_n
\delta_{\overline{\Gamma}}\,,\quad
  f_n \in \ker (I-\alpha Q_{\overline{\Gamma}} (\kappa_n)   )\,,
\end{equation}
where $\delta_{\Gamma} \in W^{-1,2}(\R^2)$ means the Dirac $\delta$ function supported by $\overline{\Gamma} $. It is convenient here to assume that the functions $f_n$ are normalized, $\|f_n\|=1$, the resulting $g_n$'s may then have, of course, non-unit lengths. Consider $\kappa_\delta >0$ such that
\begin{equation}\label{eq-defkappa}
-\kappa^2_\delta = \lambda_k (H_{\overline{\Gamma}}) -\delta^2\,.
\end{equation}
Now we return to the Birman-Schwinger argument for $H_{\overline{\Gamma}_\phi}$.  To this aim, we have to reformulate the equivalences derived in the previous section. The operator $Q_{\overline{\Gamma}_\varphi }(\kappa )$ is now determined by the kernel
$$
\mathcal{Q}_{\overline{\Gamma}_\varphi}(\kappa ; s,s')= \frac{1
}{2\pi } K_0 (\kappa |\overline{\gamma}_\varphi
(s)-\overline{\gamma}_\varphi (s')|)\,,
$$
and in the same way as before one can show that
$$
\ker (I- \alpha Q_{\overline{\Gamma}_\varphi }(\kappa_\kappa
))\neq \emptyset \quad \Leftrightarrow \quad \ker (I-
\overline{B}_\delta \overline{D}_\varphi (\kappa_\delta )) \neq
\emptyset\,,
$$
where $\kappa_\delta$ is defined by (\ref{eq-defkappa}) and
$$
\overline{D}_\varphi (\kappa ):= \alpha \left(
Q_{\overline{\Gamma}_\varphi }(\kappa ) - Q_{\overline{\Gamma}
}(\kappa )\right) \,,\quad \quad \overline{B}_\delta= (I-\alpha
Q_{\overline{\Gamma}}(\kappa_\delta  ))^{-1}\,.
$$
Note that both the $\overline{D}_\varphi (\kappa )$ and $\overline{B}_\delta$ have different asymptotics in comparison to the analogous quantities derived in previous section. Let us look into that first.

\subsubsection{Asymptotics of $\overline{D}_\varphi (\kappa )$}

\begin{lemma}
We have
\begin{equation} \overline{D}_\varphi (\kappa ) =
D^{(1)}(\kappa ) \varphi +D^{(2)}_\varphi (\kappa ) \varphi^2 \,,
\end{equation}
where $D^{(1)} (\kappa ) $ is a Hilbert-Schmidt integral operator with the  kernel
$$
\mathcal{D}^{(1)} (\kappa ; s,s') =d(s,s')
\chi_{\Omega^{\mathit{I}}}+d(s',s) \chi_{\Omega^{\mathit{II}}}\,,
$$
and\footnote{Here and in the following we use subscripts to indicate Cartesian coordinates; we put them behind the argument to avoid confusion with the wiggling parameter $\varphi$.}
$$
d(s,s'):=  -\frac{\alpha \kappa }{2\pi }\,K_0 '(\kappa
\overline{\rho } (s,s') )\,\frac{s'\overline\gamma(s)_2}{\overline{\rho }
(s,s')}\,.
$$
Moreover, the operator $D^{(2)}_\varphi (\kappa )$ is Hilbert-Schmidt as well with the HS norm uniformly bounded w.r.t. $\varphi$.
\end{lemma}
\begin{proof}
The proof employs a similar argument as that of Lemma~\ref{le-expD}. However, since the analysis requires non-trivial changes we present the reasoning in  detail. The definition of $\overline{\Gamma}_\varphi$ yields
$$
\rho (s,s')= |\overline{\gamma }_\varphi (s) - \overline{\gamma
}_\varphi(s')| = \varrho (s,s')\chi_{\Omega^{\mathit{I}}} +
\varrho (s',s)\chi_{\Omega^{\mathit{II}}}
$$
where
$$
\varrho (s,s')=  \left( \left( \overline\gamma(s)_1 -s'\cos \varphi
\right)^2 + \left( \overline\gamma(s)_2 -s'\sin \varphi \right)^2
\right)^{1/2}\,.
$$
Let $\overline{\rho} (s,s')=|\overline{\gamma } (s) - \overline{\gamma }(s')| $. Using again Taylor expansion of $\cos (\cdot)$ and $\sin (\cdot )$ and applying the mean value theorem we obtain
$$
\rho =\overline{\rho} +\left. \frac{\mathrm{d}\rho }{\mathrm{d}\varphi }\right|_0
\varphi  +\left. \zeta_\varphi (s,s')\,\varphi^2\,,\quad \quad
\zeta_\varphi (s,s')=\frac{\mathrm{d}^2\rho}{\mathrm{d}\varphi^2 }\right|_{\theta_2
\varphi}\theta_1\,,
$$
where $\theta_i \in (0, 1 )$, $\,i=1,2$. We begin with the correction term containing the first derivative. A straightforward calculation shows that
\begin{equation}\label{eq-rhobound}
\left. \frac{\mathrm{d}\rho }{\mathrm{d}\varphi }\right|_0 =
-\frac{1}{\overline{\rho} (s,s')}\left( s'\overline\gamma
(s)_2 \chi_{\Omega^{\mathit{I}}}+ s\overline\gamma
(s')_2 \chi_{\Omega^{\mathit{II}}} \right)
 \,.
\end{equation}
Suppose that $(s,s') \in \Omega$. In view of (\ref{eq-Cbound}) we obtain
\begin{equation}\label{eq-boundrho}
\overline{\rho }(s,s')\geq c |s-s'| \geq c(2|ss'|)^{1/2}\,.
\end{equation}
Using $|\overline{\rho}(s,s')|\leq |s-s'|$ and $\gamma (0)=0$ we get $ |\overline\gamma(s)_2| \leq |s|\,, $ which leads to 
$$
  \left|\frac{s'\overline\gamma(s)_2}{\overline{\rho }(s,s')} \right| \leq
c(2|ss'|)^{1/2}\,,
$$
and consequently
\begin{equation}\label{eq-estimderiv1}
\left| \left. \frac{\mathrm{d}\rho }{\mathrm{d}\epsilon }\right|_0 \right|\leq
c(2|ss'|)^{1/2}\,.
\end{equation}
We do not need to derive an explicit form of the second derivative, however, it is useful to state that it also satisfies
\begin{equation}\label{eq-estimderiv2}
\left| \zeta_\varphi (s,s')
 \right|\leq c (2|ss'|)^{1/2}\,.
\end{equation}
The remaining part of proof mimicks the argument from the proof of Lemma~\ref{le-expD}. We derive the expansion of type (\ref{eq-decompK}) with
$$
\mathcal{K}^{(1)}(s,s')=-\frac{\kappa }{2 }\, K_0 '(\kappa
\overline{\rho}(s,s'))\left( \frac{\left( s'\overline\gamma
(s)_2 \chi_{\Omega^{\mathit{I}}}+ s\overline\gamma
(s')_2 \chi_{\Omega^{\mathit{II}}} \right)}{\overline{\rho} (s,s')}
 \right)
$$
and
$$
\mathcal{K}^{(2)}_\varphi (s,s') = \kappa K_0  ' (\kappa
\overline{\rho} )\,\zeta_\varphi +\kappa ^2  K_0  '' (\kappa
(\overline{\rho}(s,s')+\check{\theta}_2 (\rho - \overline{\rho} )
) )\,\check{\theta}_1(\rho - \overline{\rho}  )^2 \,.
$$
Suppose that $(s,s')\in\Omega $. Using again the asymptotics $K_0'(x)\sim x^{-1}$ and $K_0 ''(x)\sim x^{-2}$  at the origin together with (\ref{eq-boundrho}) and the estimates (\ref{eq-estimderiv1}) and (\ref{eq-estimderiv2}), we conclude that both $\mathcal{K}^{(1)} (s,s')$ and $|\mathcal{K}^{(2)}_\varphi
(s,s')|$ are bounded for $|s-s'|\to 0$. Moreover, similarly as in (\ref{eq-boundK}) we can majorized them by exponential expression for $|s|$ and $|s'|$ large. This shows that $\mathcal{D}^{(1)} (\cdot, \cdot )= \frac{\alpha }{2\pi}\mathcal{K}^{(1)} (\cdot,\cdot )$ and $\mathcal{D}^{(2)}_\varphi  (\cdot, \cdot
)=\frac{\alpha }{2\pi }\mathcal{K}^{(2)}_\varphi (\cdot, \cdot )$ define functions from $L^{2}(\R\times \R)$. This completes the proof.
\end{proof}

\subsubsection{Asymptotics of $\overline{B}_\delta$}

We introduce the abbreviation $\mu_k =\lambda_k(H_{\overline{\Gamma }})$ for the  $k$-th non-degenerate eigenvalue of $H_{\overline{\Gamma }}$. Moreover, let $f_k$ stand for the corresponding normalized Birman-Schwinger eigenfunction, $f_k \in \ker (I-\alpha Q_{\overline{\Gamma }}(\kappa_k) )$ and $\|f_k\|=1$.

\begin{lemma} \label{le-expBII}
Let $\kappa_\delta$ be defined by (\ref{eq-defkappa}), then we have
\begin{equation}\label{eq-Bexp1}
\overline{B}_\delta =(I-\alpha Q_{\overline{\Gamma}
}(\kappa_\delta ))^{-1} =\alpha \frac{1}{ \mu_k +\kappa_\delta^2
}P_k +   P_k+W_k \,,
\end{equation}
where  $P_k = (\cdot , f_k )f_k $. Moreover, $\|W_k\| \leq c$ and $\mathrm{Ran} W_k \subset P_k^\bot$.
\end{lemma}
\begin{proof}
Note first that the operator $Q_{\overline{\Gamma}}(\kappa)$ appearing in (\ref{eq-Bexp1}) satisfies the pseudo-resolvent equivalence of the type,
$$
Q_{\overline{\Gamma} }(\kappa)= Q_{\overline{\Gamma} }(\kappa_0)
+(\kappa_0^2 - \kappa^2 )Q_{\overline{\Gamma }}
(\kappa)Q_{\overline{\Gamma }}(\kappa_0)\,,
$$
cf., e.g., \cite{Po}. This implies
$$
Q _{\overline{\Gamma }}(\kappa ) = (1 -(\kappa_0^2 - \kappa^2
)Q_{\overline{\Gamma }}(\kappa))^{-1}Q_{\overline{\Gamma
}}(\kappa_0 )
$$
and expanding the inverse into Neumann series leads to
\begin{equation}\label{eq-resoQ}
Q _{\overline{\Gamma }}(\kappa )= \sum_{n=0}^\infty(\kappa_0^2 -
\kappa^2 )^n Q_{\overline{\Gamma }}(\kappa_0 )^{n+1}
\end{equation}
provided $|\kappa_0^2 - \kappa^2|\|Q_{\overline{\Gamma }}(\kappa_0) \|<1$. In the next step we are going to show that the operator we are interested in, $I-\alpha Q_{\overline{\Gamma}}(\kappa_\delta )$, can be for small enough $\delta$ expressed as 
\begin{equation}\label{eq-BSdecomp}
I-Q_{\overline{\Gamma }}(\kappa_\delta  )= S_k - \left(
\sum_{n=1}^\infty(-\mu_k - \kappa^2_\delta  )^n  \right)P_k\,,
\end{equation}
where $S_k$ is boundedly invertible on $P_k ^\bot L^2 (\R)$. Note that the operator $Q_{\overline{\Gamma }}(\nu_k )$, where $\nu_k := \sqrt{-\mu_k}$ has a finite number of eigenvalues which we denote as  $\eta_j (Q_{\overline{\Gamma }}(\nu_k ) )$, $\,j=1,\dots,M$. We denote by $\tilde{P}_j $ the corresponding eigenprojections; then
\begin{equation}\label{eq-evdecomQ}
Q_{\overline{\Gamma }}(\nu_k ) = \sum_{i=0}^M \eta_i
(Q_{\overline{\Gamma }}(\nu_k ) ) \tilde{P}_i \,.
\end{equation}
Note that in view of (\ref{eq-BS2a}), there exists $n_0 \in \N \cup \{ 0 \}$, $n_0 \leq M$ such that $\tilde{P}_{n_0} = P_k$ and
$\eta_{n_0 } (Q_{\overline{\Gamma }}(\nu_k )) = \frac{1}{\alpha }$. Consequently, using (\ref{eq-BSdecomp}) and (\ref{eq-evdecomQ}) we get
\begin{eqnarray}\nonumber
I-\alpha Q_{\overline{\Gamma }}(\kappa_\delta ) && = I - \alpha
\sum_{i=0} ^M \sum_{n=0}^\infty (-\mu_k - \kappa_\delta^2 )^n
\eta_i
(Q_{\overline{\Gamma }}(\nu_k ))^{n+1} \tilde{P}_i   \\
\nonumber && = S_k -
\left( \sum_{n=1}^\infty
\left( \frac{-\mu_k - \kappa_\delta^2 }{\alpha}\right)^n \right) P_k\\
\label{eq-BS3} && = S_k + \frac{1}{\alpha }\, \frac{\mu_k
+\kappa_\delta^2 }{ 1+ (\mu_k +\kappa_\delta^2 )/\alpha }\, P_k\,,
\end{eqnarray}
where
$$
S_k := P_k ^{\bot }  - \alpha \sum_{i=0, \,i\neq n_0 }^M\:
\sum_{n=0}^\infty\, (-\mu_k - \kappa_\delta^2 )^n \eta_i
(Q_{\overline{\Gamma }}(\nu_k ))^{n+1} \tilde{P}_i\,.
$$
In the final step, we note that $\mathrm{Ran} \, S_k \subseteq P_k^\bot $ since $\tilde{P}_{n_0}=P_k$ holds by assumption. Furthermore, due to $\ker (I-\alpha\eta_i (Q_{\overline{\Gamma}} (\nu_k )) )=\emptyset $ the operator $S_k $ is invertible on $P_k^\bot L^{2} (\R)$ for $\delta$ small enough and its inverse is bounded. It follows from (\ref{eq-BS3}) that
$$
(I-\alpha Q_{\overline{\Gamma}} (\kappa_\delta ) )^{-1}= W_k +
\alpha\, \frac{ 1+ (\mu_k +\kappa_\delta^2 )/\alpha  }{\mu_k
+\kappa_\delta^2 }\, P_k\,,
$$
where $W_k$ is an operator acting in $P_k ^\bot L^2 (\R)$ and defined as the inverse of $S_k $ restricted  to $P_k^\bot $; this is nothing else than  the claim of the lemma.
\end{proof}
\begin{remark} \rm{Using the identity $\delta ^2 = \mu_k +\kappa_\delta ^2 $ we can rewrite (\ref{eq-Bexp}) as
\begin{equation}\label{eq-Bexp1}
B_\delta = \frac{1}{\delta^2 }P_k +  N_k \,,
\end{equation}
where $N_k:=  P_k +W_k $ has the norm uniformly bounded w.r.t. $\delta $. The above formula shows that the operator-valued function $\delta \mapsto B_\delta$ has a $\delta^{-2}$ type singularity at the point corresponding to the energy $\mu_k$. The key difference to the situation discussed in the previous section is that there we had the singularity of the type $\delta^{-1}$ for $-\kappa_\delta^2 \to -\frac14\alpha^2$, i.e. when we approached the threshold of the essential spectrum, cf.~(\ref{eq-decomB-1}) and (\ref{eq-decompoB2}). This gave rise to the different spectral aymptotics. }
\end{remark}

Now we can proceed similarly as in the previous section. Using (\ref{eq-Bexp1}) and the asymptotics of $D_\epsilon (\kappa )$
derived in Lemma~\ref{le-expD} we get
$$
\overline{B}_\delta \overline{D}_\varphi (\kappa_\delta ) = \alpha
\frac{\varphi}{\delta^2 } P_k D^{(1)} (\kappa_\delta ) + \varphi
C_\delta \,,
$$
with $\|C_\delta\|\leq c$. In analogy with (\ref{eq-final2}) the last formula can be used to demonstrate the existence of an eigenvalue for small enough $\delta$ by the Birman-Schwinger principle,
$$ \ker (I- \overline{B}_\delta \overline{D}_\varphi
(\kappa_\delta   ))\neq \emptyset\,.
$$
The invertibility of $I-\varphi C_{\delta }$ for small $\varphi$  allows us to rewrite the above condition as
\begin{equation}\label{eq-fin1}
 \ker \left(I-  \alpha
\frac{\varphi}{\delta^2 } P_k D^{(1)} (\kappa_\delta ) (I-\varphi
C_{\delta })^{-1}\right) \neq \emptyset\,.
\end{equation}
This implies
$$
\delta^2 = \alpha \varphi \mathrm{Tr} \, \left[P_k
D^{(1)}(\nu_k)\right] +o(\varphi  )=
 \alpha (D^{(1)}(\nu_k)f_k , f_k ) \varphi  +o(\varphi )\,,
$$
and thus finally
$$
\lambda _k (H_{\overline{\Gamma}_\varphi}) = \lambda _k
(H_{\overline{\Gamma}})-  \alpha (D^{(1)}(\nu_k)f_k , f_k )
\varphi  +o(\varphi  )\,.
$$

\subsection{The degenerate case }

Consider finally a degenerate eigenvalue $\mu$ of $H_{\overline{\Gamma}}$ of multiplicity $k$, i.e. $\mu = \lambda _k(H_{\overline{\Gamma}})$ for $k=j+1,\dots,j+m$. Furthermore, let $\{f_k\}_{k=j+1}^{j+m}$ stand for the set of the corresponding normalized eigenvectors. Repeating the steps made in the proof of Lemma~\ref{le-expBII} we get
\begin{lemma} \label{le-expBII-degen}
We have
\begin{equation}\label{eq-Bexp}
\overline{B}_\delta =\alpha  \sum_{k=j+1}^{j+m}\frac{1}{ \mu
+\kappa_\delta^2 }\,P_k +   V_k \,,
\end{equation}
where  $P_k = (\cdot, f_k )f_k$ is the corresponding eigenprojection, and moreover, $\|V_k\| \leq c$.
\end{lemma}
Now we can proceed in the same way as in the non-degenerate case. Using
(\ref{eq-Bexp}) we derive the equivalence analogous to (\ref{eq-fin1})  with the obvious modifications
$$
\ker \left( \delta^2 - \alpha \varphi  \sum_{k=j+1}^{j+m}  \, P_k
D^{(1)}(\kappa_k)(I-\varphi C_\delta )^{-1} \right) \neq \emptyset
\,.
$$
Using again the expansion  of $(I-\varphi C_\delta )^{-1}$ and the
fact that $\{f_k\}_{k=j+1}^{j+m}$ can be chosen as an orthonormal set we get
$$
\lambda _k (H_{\overline{\Gamma}_\varphi}) = \lambda
(H_{\overline{\Gamma}})-  \alpha (D^{(1)}(\kappa)f_k , f_k )
\varphi  +o(\varphi  )\,.
$$
This completes  the proof.

\subsection*{Acknowledgements}

The research was supported by the Czech Science Foundation (GA\v{C}R) within the project 14-06818S and the Polish  Nation Science Centre within the project DEC-2013/11/B/ST1/03067. S.K. thanks the Department of Theoretical Physics, NPI CAS in Re\v{z}, for the hospitality in December 2014, when this work was started.

\end{document}